\documentclass[
]{ceurart}

\usepackage{amsthm}
\usepackage{graphicx}
\usepackage{textcomp}

\usepackage[utf8]{inputenc} 
\usepackage[T1]{fontenc}    
\usepackage{hyperref}       
\usepackage{url}            
\usepackage{booktabs}   
\usepackage{subcaption}
\usepackage{adjustbox} 
\usepackage{array}
\usepackage{color,soul}

\usepackage{bm}

\usepackage[]{algorithm}
\usepackage{algpseudocode}

\newtheorem{definition}{Definition}
\newtheorem{lemma}{Lemma}
\newtheorem{theorem}{Theorem}

\renewcommand{\P}{\mathrm{Pr}}

\usepackage{mathtools}
\DeclarePairedDelimiter\ceil{\lceil}{\rceil}
\DeclarePairedDelimiter\floor{\lfloor}{\rfloor}
\DeclarePairedDelimiter\abs{\lvert}{\rvert}   

\usepackage{listings}
\begin{document}

\copyrightyear{2025}
\copyrightclause{Copyright for this paper by its authors.
  Use permitted under Creative Commons License Attribution 4.0
  International (CC BY 4.0).}

\conference{}

\title{Space-Efficient Private Estimation of Quantiles}


\author[1]{Massimo Cafaro}[%
orcid=0000-0003-1118-7109,
email=massimo.cafaro@unisalento.it,
]
\cormark[1]
\fnmark[1]

\author[1]{Angelo Coluccia}[%
orcid=0000-0001-7118-9734,
email=angelo.coluccia@unisalento.it,
]
\fnmark[1]

\author[1]{Italo Epicoco}[%
orcid=0000-0002-6408-1335,
email=italo.epicoco@unisalento.it,
]
\fnmark[1]

\author[1]{Marco Pulimeno}[%
orcid=0000-0002-4201-1504,
email=marco.pulimeno@unisalento.it,
]
\fnmark[1]

\address[1]{Department of Engineering for Innovation, University of Salento, Lecce, 73100 Italy}

\cortext[1]{Corresponding author.}
\fntext[1]{These authors contributed equally.}

\begin{abstract}
Fast and accurate estimation of quantiles on data streams coming from communication networks, Internet
of Things (IoT), and alike, is at the heart of important data processing applications including statistical analysis, latency monitoring, query optimization for parallel database management systems, and more. Indeed, quantiles are more robust indicators for the underlying distribution, compared to moment-based indicators such as mean and variance. The streaming setting additionally constrains accurate tracking of quantiles, as stream items may arrive at a very high rate and must be processed as quickly as possible and discarded, being their storage usually unfeasible. Since an exact solution is only possible when data are fully stored, the goal in practical contexts is to provide an approximate solution with a provably guaranteed bound on the approximation error committed, while using a minimal amount of space. At the same time, with the increasing amount of personal and sensitive information exchanged, it is essential to design privacy protection techniques to ensure confidentiality and data integrity. In this paper we present the following differentially private streaming algorithms for frugal estimation of a quantile: \textsc{DP-Frugal-1U-L}, \textsc{DP-Frugal-1U-G}, \textsc{DP-Frugal-1U-$\rho$}. Frugality refers to the ability of the algorithms to provide a good approximation to the sought quantile using a modest amount of space, either one or two units of memory. We provide a theoretical analysis and  experimental results.
\end{abstract}

\begin{keywords}
Quantiles \sep
Streaming \sep
Differential Privacy \sep
Frugal Algorithms
\end{keywords}

\maketitle

\section{Introduction}
\label{introduction}

Data processing of network-originated streams, coming from wireless sensors, telecommunications networks, Internet of Things (IoT), cyber-physical systems, and many other sources, is ubiquitously present in contemporary ICT applications.
With the increasing amount of personal and sensitive information exchanged, it is essential to design privacy protection techniques to ensure confidentiality and data integrity. At the same time, there is a trend in trying to reduce the computational complexity and memory requirements of data processing algorithms, with the goal of lowering costs and environmental impact. Such issues are expected to play an even more important role in  upcoming 6G wireless networks, with a large amount of data  generated and processed at the edge of the network in resource-constrained devices, and potentially exposed to security and privacy attacks \cite{6G}. 
For such reasons, recent work is actively working to obtain both i) privacy-preserving and/or ii) reduced-complexity versions of classical algorithms. 

As to the first requirement, differential privacy (DP) has attracted significant interest among various research communities, including computer science, communications, data and signal processing \cite{IoT21,DPtutorial_ACM,DP_SPM}.
Among the algorithms that have been recently revisited under the DP paradigm we can cite principal component analysis (PCA) \cite{PCA}, federated learning and microaggragation for IoT \cite{IoT22,IoT23}, and various algorithms for the estimation of mean values or other ensemble statistics under different settings, e.g., \cite{NEURIPS2021,mean,Scaglione}. The interested reader may refer to one of the many surveys available; we recall here \cite{survey}, which is a comprehensive and recent survey discussing differentially private algorithms for both data publishing and data analysis.  

As to the second requirement, while the processing of clipped or quantized data is a classical topic \cite{clipped_noise,quantization}, recent research is focusing on extreme settings where only a very limited amount of information is retained. For instance, since as known the complexity of analog-to-digital conversion (ADC) grows exponentially with the number of bits, and the power dissipated by the ADC circuit scales approximately at the same pace \cite{ADC_JSAC}, it is very convenient to adopt coarse quantization or even a single comparator that forwards only the sign of the  signal, discarding all the remaining information. While this brings enormous advantages, the significant information loss requires more advanced algorithms for the processing \cite{Eldar1, Eldar2}.
The same principle has been applied to the processing downstream the ADC, specifically through the introduction of binning as well as memory-constrained algorithms. In particular, a recent trend is the adoption of computational approaches that require only a few or a single memory variables. This strategy,  termed ``frugal'', allows the processing of large amount of data in resource-constrained devices, and is therefore a very timely research direction.

The streaming setting adds additional constrains, since stream items may arrive at a very high rate and must be processed as quickly as possible and discarded \cite{IoT24}, owing to the fact that storing them is usually unfeasible given that the input stream may be potentially infinite. An exact solution is only possible if all of the stream items are stored, so that, since streaming algorithms strive to use a minimal amount of space, the goal is to provide an approximate solution with a provably guaranteed bound on the approximation error committed.

Inspired by the  above, this work aims at obtaining DP quantile estimation algorithms with very low memory requirement, so fulfilling both  i) and ii) at the same time.
As known, quantiles (e.g., median, quartiles, percentiles, etc.) are a key tool in robust statistics \cite{huber2009}, whose aim is to obtain ensemble information from a set of data, namely an estimate of a parameter of interest, while limiting the impact of outliers or heavy tails. 
Indeed, estimating an unknown parameter from a set of random variables, which may also arise from observations at different sources \cite{tsai2011cooperative,nevat2013random} is typically performed by computing some average  on a data stream. 
However, data heterogeneity (including non-independent and identically distributed distributions \cite{IoTNIID}) may arise due to local causes~\cite{fijalkow2016parameter,chen2019resilient}, such as the presence of outliers~\cite{papageorgiou2015robust},  
 heavy-tails, random effects with heterogeneous variance and/or variable sample size, for which robust tools not requiring knowledge of the data distribution are needed \cite[and references therein]{Grassi2023}. Heavy tails are found in many types of data, including financial \cite{muller1998heavy}, natural and computer-originated data \cite{machado2015review}, a prominent example of the latter being the   Internet \cite{Internet3}: in such a context, the use of quantiles/percentiles is indeed a popular practical tool to cope with the wild variability of IP traffic \cite{Internet3,willinger2004more}, e.g.  
 for robust estimation of network-wide key   performance  indicators (KPIs) such as one-way delays \cite{OWD} and round-trip times \cite{RTT}. 
 Besides these contexts, fast and accurate tracking of quantiles in a streaming setting is of utmost importance also in database query optimizers, data splitting for parallel computation in database management systems, etc. 

While recent work has provided DP algorithms for mean values  \cite{NEURIPS2021,mean}, to the best of our knowledge no DP algorithm is available in the literature for quantile estimation via frugal computation. We base our work on the \textsc{Frugal-1U} algorithm \cite{frugal}, discussed  in Section \ref{frugal-1u-alg}, and present preliminary results. Overall, we provide the following original contributions, without assuming knowledge of the data distribution:

\begin{itemize}
    \item we analyze the \textsc{Frugal-1U} algorithm and prove that its global sensitivity is bounded and equal to 2; next, we design three DP versions of the algorithm based respectively on the Laplace mechanism, the Gaussian mechanism and on $\rho$ zero-concentrated DP;
   
    \item we validate the theoretical results through  simulations.
\end{itemize}

The rest of this paper is organized as follows. Section \ref{notation} provides necessary definitions and notation used throughout the manuscript. Section \ref{frugal-1u-alg} introduces the \textsc{Frugal-1U} algorithm whilst Section \ref{dp-frugal-1u} presents our analysis and three corresponding DP algorithms. We present the experimental results in section \ref{results} and draw our conclusions in Section \ref{conclusions}.

\section{Preliminary Definitions and Notation}
\label{notation}
In this Section, we briefly recall the definitions and notations that shall be used throughout this paper. We begin by giving a formal definition of quantiles.

\begin{definition} (Lower and upper $q$-quantile) Given a set $A$ of size $n$ over $\mathbb{R}$, let $R(x)$ be the rank of the element $x$, i.e., the number of elements in $A$ smaller than or equal to $x$. Then, the lower (respectively upper) $q$-quantile  item $x_q \in A$ is the item $x$ whose rank $R(x)$ in the sorted set $A$ is $\floor{1+q(n-1)}$ (respectively $\ceil{1+q(n-1)}$) for  $0 \leq q \leq 1$. 
\end{definition}

The accuracy related to the estimation of a quantile can be defined either as rank or relative accuracy. In this paper, we deal with algorithms that provide rank accuracy, which is defined as follows.

\begin{definition} (Rank accuracy) For all items $v$ and a given tolerance $\epsilon$, return an estimated rank $\tilde{R}$ such that $\abs{\tilde{R}(v) - R(v)} \leq \epsilon n$.
\end{definition}

Next, we introduce the main concepts underlying DP. We focus on the so-called \textit{central model} of DP. Actually, two definitions are possible, as follows.

\begin{definition} \label{udf} (Unbounded differential privacy, also known as the add-remove model \cite{10.1007/11681878_14} \cite{TCS-042}) 
Two datasets $x$ and $x^{\prime}$ are considered neighbors if $x^{\prime}$ can be obtained from $x$ by adding or removing one row. Under unbounded DP, the sizes of $x$ and $x^{\prime}$  are different (by one row): $\left|x \backslash x^{\prime}\right|+\left|x^{\prime} \backslash x\right|=1$. 
\end{definition}

\begin{definition} \label{bdp} (Bounded differential privacy, also known as the swap or the update/replace model \cite{10.1007/11681878_14} \cite{Vadhan2017}) 
Two datasets $x$ and $x^{\prime}$ are considered neighbors if $x^{\prime}$ can be obtained from $x$ by changing one row. Under bounded DP, the sizes of $x$ and $x^{\prime}$ are equal: $\left|x \backslash x^{\prime}\right|=1 \text { and }\left|x^{\prime} \backslash x\right|=1$.
\end{definition}

In this paper, we adopt bounded DP. Next, we define $\epsilon$-DP. 

\begin{definition} \label{pdp} ($\epsilon$-differential privacy) A function which satisfies DP is called a mechanism; we say that a mechanism $\mathcal{F}$ satisfies pure DP if for all neighboring datasets  $x$ and $x^{\prime}$ and all possible sets of outputs $\mathcal{S}$, it holds that $\frac{\operatorname{Pr}[\mathcal{F}(x) \in \mathcal{S}]}{\operatorname{Pr}\left[\mathcal{F}\left(x^{\prime}\right) \in \mathcal{S} \right]} \leq e^\epsilon$. The $\epsilon$ parameter in the definition is called the privacy parameter or privacy budget.
\end{definition}

The $\epsilon$ parameter is strictly related to the desired amount of privacy. In practice, there is trade-off going on, since smaller values of this parameter provide higher privacy but at the cost of less \emph{utility} and vice-versa. In this context, utility refers to the possibility of using the obtained results for further investigations, namely statistical analyses. Therefore, the trade-off may be understood considering that setting $\epsilon$ to a small value require the mechanism $\mathcal{F}$ to provide very similar outputs when instantiated on similar inputs (so, higher privacy, obtained by injecting more noise which in turn undermines utility); on the contrary, a large value provides less similarity of the outputs (so, less privacy but increased utility). Besides pure DP, a different notion, called approximate (or, alternatively, $\epsilon, \delta$) DP, is also available. 

\begin{definition} \label{adp} ($(\epsilon, \delta)$-differential privacy) A mechanism $\mathcal{F}$ satisfies ($\epsilon, \delta$)-DP if for all neighboring datasets  $x$ and $x^{\prime}$ and all possible sets of outputs $\mathcal{S}$, it holds that $\operatorname{Pr}[\mathcal{F}(x) \in \mathcal{S}] \leq e^\epsilon \operatorname{Pr}\left[\mathcal{F}\left(x^{\prime}\right) \in \mathcal{S} \right]+\delta$, where the privacy parameter $0 \leq \delta < 1$ represents a failure probability.
\end{definition}

The definition implies that (i) with probability $1-\delta$ it holds that $\frac{\operatorname{Pr}[\mathcal{F}(x) \in \mathcal{S}]}{\operatorname{Pr}\left[\mathcal{F}\left(x^{\prime}\right) \in \mathcal{S} \right]} \leq e^\epsilon$ and (ii) with probability $\delta$ no guarantee holds. As a consequence, $\delta$ is required to be very small.

In order to define a mechanism, we need to introduce the notion of sensitivity. In practice, the sensitivity of a function reflects the amount the function's output will change when its input changes. Formally, given the universe of datasets, denoted by $\mathcal{D}$, the sensitivity  of a function $f$, called \emph{global sensitivity}, is defined as follows.

\begin{definition} \label{sensitivity} (Global sensitivity) Given a function $f: \mathcal{D} \rightarrow \mathbb{R}$ mapping a dataset in $\mathcal{D}$ to a real number, the global sensitivity of $f$ is $G S(f)=\max _{x, x^{\prime}: d\left(x, x^{\prime}\right)\leq 1}\left|f(x)-f\left(x^{\prime}\right)\right|$, where $d\left(x, x^{\prime}\right)$ represents the distance between two datasets  $x$, $x^{\prime}$.
\end{definition}

We now define two mechanisms, respectively the Laplace and the Gaussian mechanism. The former must be used with pure DP, the latter with approximate DP. 

\begin{definition} \label{laplace} (Laplace mechanism) Given a function $f:~\mathcal{D}~\rightarrow~\mathbb{R}$ mapping a dataset in $\mathcal{D}$ to a real number, $\mathcal{F}(x)=f(x)+\operatorname{Lap}\left(\frac{s}{\epsilon}\right)$ satisfies $\epsilon$-DP. $Lap(S)$ denotes sampling from the Laplace distribution with center 0 and scale $S$, whilst $s$ is the sensitivity of $f$.
\end{definition}

\begin{definition} \label{gaussian} (Gaussian mechanism) Given a function $f:~\mathcal{D}~\rightarrow~\mathbb{R}$ mapping a dataset in $\mathcal{D}$ to a real number, $\mathcal{F}(x)=f(x)+\mathcal{N}\left(\sigma^2\right)$ satisfies $(\epsilon, \delta)$-DP, where $\sigma^2=\frac{2 s^2 \ln (1.25 / \delta)}{\epsilon^2}$ and $s$ is the sensitivity of $f$. $\mathcal{N}\left(\sigma^2\right)$ denotes sampling from the Gaussian (normal) distribution with center 0 and variance $\sigma^2$. 
\end{definition}

The Gaussian mechanism also satifies a stronger notion of privacy, known as $\rho$ zero-concentrated differential privacy ($\rho$-zCDP); its definition uses a single privacy parameter $\rho$, and lies between pure and approximate DP. Moreover, $\rho$-zCDP has been shown to be equivalent (i.e., it can be translated) to standard notions of privacy.

\begin{definition} \label{zCDP} ($\rho$-zCDP) A mechanism $\mathcal{F}$ satisfies zero-concentrated DP if for all neighboring datasets  $x$ and $x^{\prime}$ and all $\alpha \in (1, \infty)$, it holds that $D_\alpha\left(\mathcal{F}(x) \| \mathcal{F}\left(x^{\prime}\right)\right) \leq \rho \alpha$, where $D_\alpha(P \| Q)=\frac{1}{\alpha-1} \ln E_{x \sim Q}\left(\frac{P(x)}{Q(x)}\right)^\alpha$ is the Rényi divergence.
\end{definition}

It can be shown that $\rho$-zCDP can be converted to $(\epsilon, \delta)$-DP as follows. If the mechanism $\mathcal{F}$ satisfies $\rho$-zCDP, then for $\delta > 0$ it also satisfies $(\epsilon, \delta)$-differential privacy for $\epsilon=\rho+2 \sqrt{\rho \log (1 / \delta)}$. Moreover the Gaussian mechanism can be adapted to work with $\rho$-zCDP as follows.

\begin{definition} \label{rzCDP-gaussian} ($\rho$-zCDP Gaussian mechanism) Given a function $f:~\mathcal{D}~\rightarrow~\mathbb{R}$ mapping a dataset in $\mathcal{D}$ to a real number, $\mathcal{F}(x)=f(x)+\mathcal{N}\left(\sigma^2\right) \text { where } \sigma^2=\frac{s^2}{2 \rho}$ satisfies $\rho$-zCDP, where $s$ is the sensitivity of $f$.
\end{definition}

We briefly introduce the concept of utility, which quantifies how much a DP result is useful for a subsequent data analysis. Therefore, the analysis to be performed plays a key role here, since DP results affected by a significant error may or may not be useful to the analyst. One way to overcame the dependence from the analysis is the use of the related concept of accuracy, which is the distance between the true value computed without DP and the DP released value. Therefore, accuracy is often used in place of utility, because more accurate results are generally more useful for an analysis. The so-called $(\alpha, \beta)$-accuracy framework \cite{DBS-066} can be used to measure accuracy. Here, $\alpha$ represents an upper bound on the absolute error committed, whilst $\beta$ is the probability to violate this bound.

\begin{definition} \label{accuracy} ($(\alpha, \beta)$-accuracy) Given a function $f: \mathcal{D} \rightarrow \mathbb{R}$ mapping a dataset $x \in \mathcal{D}$ to a real number, and a DP mechanism $\mathcal{M}_f: \mathcal{D} \rightarrow \mathbb{R}$, $\mathcal{M}_f$ is $(\alpha, \beta)$-accurate if $\operatorname{Pr}\left[\left\|f(x)-\mathcal{M}_{f(x)}\right\|_{\infty} \geq \alpha\right] \leq \beta$.	
\end{definition}

It can be shown \cite{DBS-066}, starting from the Cumulative Distribution Function for the Laplace distribution $\operatorname{Lap}(0,b)$, that the Laplace mechanism is $(\alpha, \beta)$-accurate with
\begin{equation}
\label{alpha}
\alpha=\ln \left(\frac{1}{\beta}\right) \cdot\left(\frac{s}{\epsilon}\right).
\end{equation}

Regarding the Gaussian and the $\rho$-zCDP mechanisms, we did not find in the literature a corresponding derivation for the $\alpha$ value; as an additional contribution, here we derive their analytical form. We start by considering the Cumulative Distribution Function for the normal distribution $\mathcal{N}\left(\sigma\right)$, which is $\frac{1}{2}\left[\operatorname{erfc}\left(-\frac{x}{\sigma \sqrt{2}}\right)\right]$. The probability $\operatorname{Pr}\left[ X > x \right]$ is $1 - \frac{1}{2}\left[\operatorname{erfc}\left(-\frac{x}{\sigma \sqrt{2}}\right)\right]$, so that, substituting $x = t \sigma$, we get:
\begin{equation}
\operatorname{Pr}\left[ X > x \right] = 1-\frac{1}{2} \operatorname{erfc}\left[-\frac{t}{\sqrt{2}}\right].
\end{equation}

Therefore, we need to solve, taking into account that $0 < \beta < 1$, the following equation, with regard to $t$:
\begin{equation}
1-\frac{1}{2} \operatorname{erfc}\left[-\frac{t}{\sqrt{2}}\right] \leq \beta
\end{equation}
obtaining 
\begin{equation}
	t \geq -\sqrt{2} \ \operatorname{erfc}^{-1}(2 (1 - \beta)).
\end{equation}
\noindent It follows that the Gaussian mechanism is $(\alpha, \beta)$-accurate with
\begin{equation}
\label{alpha-gauss}
\alpha = (-\sqrt{2} \ \operatorname{erfc}^{-1}(2 (1 - \beta))) \sqrt{\frac{2 s^2 \ln \left(\frac{1.25}{\delta}\right)}{\epsilon^2}}.
\end{equation}

Reasoning as before, we can also derive that the $\rho$-zCDP mechanism is $(\alpha, \beta)$-accurate with
\begin{equation}
\label{alpha-rho}
\alpha = (-\sqrt{2} \ \operatorname{erfc}^{-1}(2 (1 - \beta))) \sqrt{\frac{s^2}{2 \rho}}.
\end{equation}

Next, we introduce the \textsc{Frugal-1U} algorithm.

\section{The \textsc{Frugal-1U} Algorithm}
\label{frugal-1u-alg}
Among the many algorithms that have been designed for tracking quantiles in a streaming setting, \textsc{Frugal} \cite{frugal} besides being fast and accurate, also restricts by design the amount of memory that can be used. It is well-known that in the streaming setting the main goal is to deliver a high-quality approximation of the result (this may provide either an additive or a multiplicative guarantee) by using the lowest possible amount of space. In practice, there is a tradeoff between the amount of space used by an algorithm and the corresponding accuracy that can be achieved. Surprisingly, \textsc{Frugal} only requires one unit of memory to track a quantile. The authors of \textsc{Frugal} have also designed a variant of the algorithm that uses two units of memory. In this Section, we introduce the one unit of memory version, which is called \textsc{Frugal-1U}. Algorithm \ref{Frugal-1U} provides the pseudo-code for \textsc{Frugal-1U}.

\begin{algorithm}
	\caption{Frugal-1U}
		\label{Frugal-1U}
	 \begin{algorithmic}[1]
		\Require{Data stream $S$, quantile $q$, one unit of memory $\tilde{m}$}
		\Ensure{estimated quantile value $\tilde{m}$}
		\State $\tilde{m} = 0$
		\ForAll {$s_i \in S$}
			\State $rand = random(0,1)$
			\If {$s_i > \tilde{m}$ and $rand > 1-q$}
				\State $\tilde{m} = \tilde{m} + 1$
			\ElsIf {$s_i < \tilde{m}$ and $rand > q$}				\State $\tilde{m} = \tilde{m} - 1$
			\EndIf
		\EndFor
		\State \Return $\tilde{m}$
	\end{algorithmic}
\end{algorithm}

The algorithm works as follows. First, $\tilde{m}$ is initialized to zero (however, note that it can be alternatively initialized to the value of the first stream item, in order to increase the speed of convergence of the estimate towards the value of the true quantile). This variable will be dynamically updated each time a new item $s_i$ arrives from the input stream $S$, and its value represents the estimate of the quantile $q$ being tracked. The update is quite simple, since it only requires $\tilde{m}$ to be increased or decreased by one. Specifically, a random number $0 < rand < 1$ is generated by using a pseudo-random number generator (the call $random(0,1)$ in the pseudo-code) and if the incoming stream item is greater than the estimate $\tilde{m}$ and $rand > 1-q$, then the estimate $\tilde{m}$ is increased, otherwise if the incoming stream item is smaller than the estimate $\tilde{m}$ and $rand > q$, then the estimate $\tilde{m}$ is decreased.  Obviously, the algorithm is really fast and can process an incoming item in worst-case $O(1)$ time. Therefore, a stream of length $n$ can be processed in worst-case $O(n)$ time and $O(1)$ space. 

Despite its simplicity, the algorithm provides good accuracy, as shown by the authors. The proof is challenging since the algorithm's analysis is quite involved. The complexity in the worst case is $O(n)$, since $n$ items are processed in worst case $O(1)$ time. 

Finally, the algorithm has been designed to deal with an input stream consisting of integer values distributed over the domain $[N] = \{1,2,3, \ldots, N\}$. This is not a limitation though, owing to the fact that one can process a stream of real values as follows: fix a desired precision, say three decimal digits, then each incoming stream item with real value can be converted to an integer by multiplying it by $10^3$ and then truncating the result by taking the floor. If the maximum number of digits following the decimal point is known in advance, truncation may be avoided altogether: letting $m$ by the maximum number of digits following the decimal point, it suffices to multiply by $10^m$. Obviously, the estimated quantile may be converted back to a real number dividing the result by the fixed precision selected or by $10^m$.

\section{Differentially-Private \textsc{Frugal-1U}}
\label{dp-frugal-1u}
In this Section, we analyze the \textsc{Frugal-1U} algorithm and design  DP variants of it. As shown in Section \ref{frugal-1u-alg}, the algorithm is quite simple. In order to estimate a quantile $q$, the current estimate $\tilde{m}$ is either incremented or decremented by one based on the value of the incoming stream item $s_i$. The increments are applied with probability $q$ and the decrements with probability $1-q$. 

Our DP versions of the algorithm are based on the definition of bounded DP (see Definition \ref{bdp}), in which two datasets $x$ and $x^\prime$ are considered neighbors if $x^\prime$ can be obtained from $x$ by changing one row. Owing to our choice, we need to analyze the impact of changing one incoming stream item with a different one on the quantile estimate $\tilde{m}$. The following Lemma is our fundamental result to then obtain DP versions of Frugal-1U.

\begin{lemma}
\label{frugal-1u-sensitivity}
Under bounded DP, the global sensitivity of the \textsc{Frugal-1U} algorithm is 2.
\end{lemma}

\begin{proof}
 Let $s_i$ be the item to be changed, and $s_j \neq s_i$ the item replacing $s_i$. There are a few symmetric cases to consider. Let $s_i$ be the $i$-th stream item, so that the length of the stream $S$ is equal to $i-1$ before the arrival of $s_i$ and equal to $i$ immediately after. Moreover, denote by $\tilde{m}_{i-1}$ the estimate of the quantile $q$ before the arrival of $s_i$ and by $\tilde{m}_i$ after seeing the item $s_i$. Suppose that the arrival of $s_i$ causes $\tilde{m}_i$ to increase by one with regard to $\tilde{m}_{i-1}$, i.e., $\tilde{m}_i = \tilde{m}_{i-1} + 1$. Substituting $s_i$ with $s_j$ therefore can lead to the following cases: either $\tilde{m}_i = \tilde{m}_{i-1} - 1$ or $\tilde{m}_i = \tilde{m}_{i-1} + 1$. Therefore, the estimate is unchanged or it is increased by 2. Similarly, assuming that the arrival of $s_i$ causes $\tilde{m}_i$ to decrease by one with regard to $\tilde{m}_{i-1}$, i.e., $\tilde{m}_i = \tilde{m}_{i-1} - 1$, then there are, symmetrically, the following cases: either $\tilde{m}_i = \tilde{m}_{i-1} + 1$ or $\tilde{m}_i = \tilde{m}_{i-1} - 1$. Therefore, the estimate is unchanged or is decremented by 2.   It follows that the global sensitivity of the algorithm is $\max _{x, x^{\prime}: d\left(x, x^{\prime}\right)\leq 1}\left|f(x)-f\left(x^{\prime}\right)\right| = 2$.
\end{proof}

Since the global sensitivity is 2, \textsc{DP-Frugal-1U-L}, a pure DP (see Definition \ref{pdp}) variant of the algorithm can be obtained by using the Laplace mechanism. We are now in the position to state the following theorem.

\begin{theorem}
\textsc{Frugal-1U} can be made $\epsilon$-DP by adding to the quantile estimate returned by the algorithm noise sampled from a Laplace distribution  with center 0 and scale $S$ as follows: $\tilde{m} = \tilde{m} + Lap(\frac{2}{\epsilon})$. 
\end{theorem}
 
\begin{proof}
It follows straight from Lemma \ref{frugal-1u-sensitivity} and Definition \ref{laplace}.
\end{proof}

Next, we design \textsc{DP-Frugal-1U-G}, a $(\epsilon, \delta)$-DP (see Definition \ref{adp}) version of the algorithm, by using the Gaussian mechanism.

\begin{theorem}
\textsc{Frugal-1U} can be made $(\epsilon, \delta)$-DP by adding to the quantile estimate returned by the algorithm noise sampled from a Gaussian distribution as follows: $\mathcal{F}(x)=f(x)+\mathcal{N}\left(\sigma^2\right)$ where $\sigma^2=\frac{8 \ln (1.25 / \delta)}{\epsilon^2}$. 
\end{theorem}

\begin{proof}
It follows straight from Lemma \ref{frugal-1u-sensitivity} and Definition \ref{gaussian}.
\end{proof}

Finally, we design \textsc{DP-Frugal-1U-$\rho$}, a $\rho$-zCDF version of the algorithm.

\begin{theorem}
\textsc{Frugal-1U} can be made $\rho$-zCDF by adding to the quantile estimate returned by the algorithm noise sampled from a Gaussian distribution as follows: $\mathcal{F}(x)=f(x)+\mathcal{N}\left(\sigma^2\right) \text { where } \sigma^2=\frac{2}{\rho}$. 
\end{theorem}

\begin{proof}
It follows straight from Lemma \ref{frugal-1u-sensitivity} and Definition \ref{rzCDP-gaussian}.
\end{proof}

Finally, we remark here that, contrary to many DP algorithms that initialize a data structure or a sketch using suitable noise, it is not possile to initialize the quantile estimate of \textsc{Frugal-1U} using the noise required by one of the proposed mechanisms. The reason is two-fold. First, the algorithm processes integer items, so that its initial estimate must be an integer as well whilst, in general, the noise is a floating point value. But, even assuming that we initialize the estimate to an integer value related to the noise (perhaps taking its floor or the ceiling), this will not help in any way since, by design, the algorithm adapts dynamically to the observed input items and converges to the estimated quantile. Therefore, the second reason is that the noise added will be silently discarded by the algorithm when converging to the quantile estimate. As a consequence, the noise must be added only after the algorithm termination to the returned estimated quantile. 

\section{Experimental Results}
\label{results}

 \begin{table}
 	\caption{Synthetic datasets.} 
	\label{datasets}
	\centering
	\begin{tabular}{llll}
		\textbf{Dataset} & \textbf{Distribution}& \textbf{Parameters}& \textbf{PDF}\\
		\hline
		D1& Uniform &  $[0, 1000]$ & \includegraphics[width=2cm]{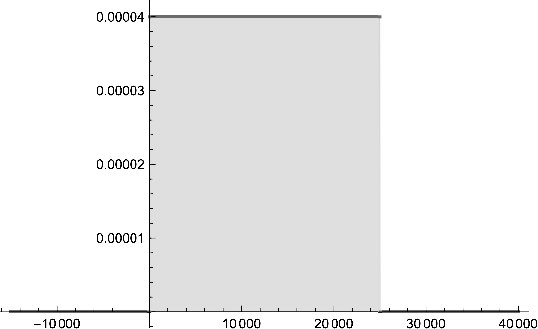}\\
		D2 & Chi square &   $\alpha = 5$ & \includegraphics[width=2cm]{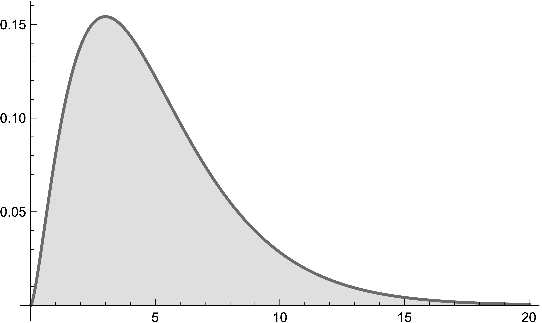}\\
		D3 & Exponential  & $\alpha = 0.5$ & \includegraphics[width=2cm]{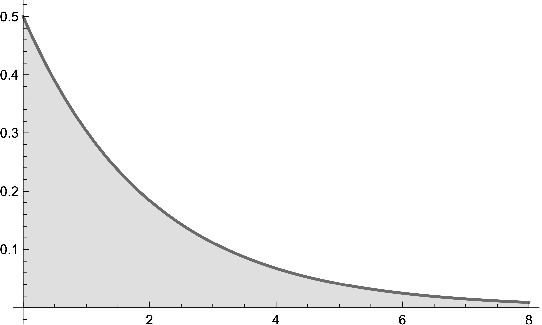}\\
		D4 & Lognormal & $\alpha =1, \beta = 1.5$ & \includegraphics[width=2cm]{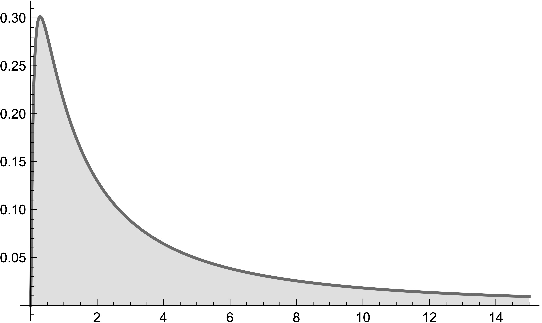}\\
		D5 & Normal & $\mu = 50, \sigma = 2$ & \includegraphics[width=2cm]{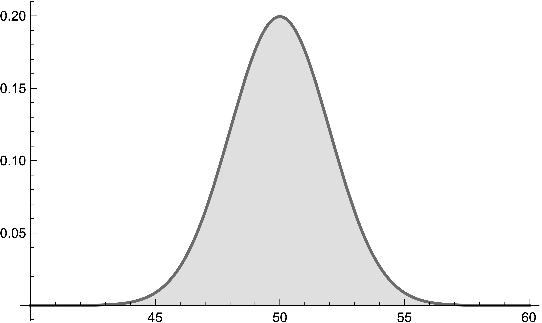}\\
		D6 & Cauchy & $\alpha = 10000, \beta = 1250$ & \includegraphics[width=2cm]{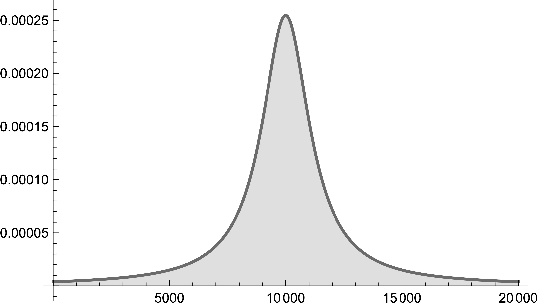}\\
		D7 & Extreme Value & $\alpha = 20, \beta = 2$ & \includegraphics[width=2cm]{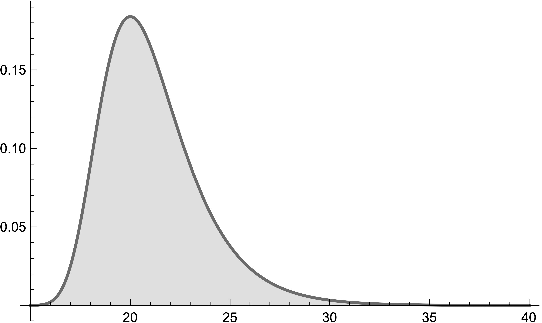}\\
		D8 & Gamma & $a = 2, b = 4$ & \includegraphics[width=2cm]{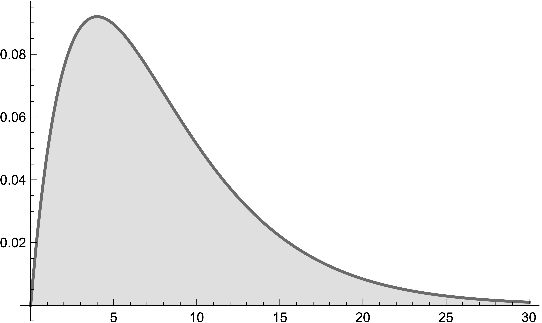}\\
		\hline
	\end{tabular}
\end{table}

 In this section we present and discuss the results of the experiments carried out for \textsc{Frugal-1U-L}, \textsc{Frugal-1U-G}, \textsc{Frugal-1U-$\rho$}. The experiments are limited to our algorithms owing to the fact that, to the best of our knowledge, there are no other frugal algorithms for the estimation of quantiles designed for the central model of differential privacy.  
 The source code has been compiled using the Apple clang compiler v15.0 with the following flags: -Os -std=c++14 (it is worth recalling  that on macOS the flag -Os optimizes for size and usually on this platform the resulting executable is faster than the executable obtained by compiling using the -O3 flag). The tests have been carried out on an Apple MacBook Pro laptop equipped with 64 GB of RAM and a 2.3 GHz 8-core Intel Core i9. The experiments have been repeated ten times for each specific category of test and the results have been averaged; the seeds used to initialize the pseudo-random generators are the same for each  experiment and algorithm being tested.
 
The tests have been performed on eight synthetic datasets, whose properties are summarized in Table \ref{datasets}. The experiments have been executed varying the stream length, the quantile, the privacy budget, $\epsilon$, $\delta$ and $\rho$. Table \ref{params} reports the default settings for the parameters. 

\begin{table}
	\caption{Default settings of the parameters.} 
	\label{params}
	\centering
	\begin{tabular}{lll}
		\textbf{Parameter} & \textbf{Values} & \textbf{Default}\\
		\hline
		quantile &   $\{0.1, 0.3, 0.5, 0.99\}$ & 0.99\\
		stream length & $\{ 10M, 50M, 75M, 100M \}$ & $10M$\\
		$\epsilon$ &   $\{0.1, 0.5, 1, 2\}$ & 1\\
		$\delta$ &   $\{0.01, 0.04, 0.08, 0.1\}$ & 0.04\\
		$\rho$ &   $\{0.1, 0.5, 1, 5\}$ & 1\\
		\hline
	\end{tabular}
\end{table}

\begin{figure*}
    \centering
    \begin{subfigure}{0.33\textwidth}
        \includegraphics[width=0.99\linewidth]{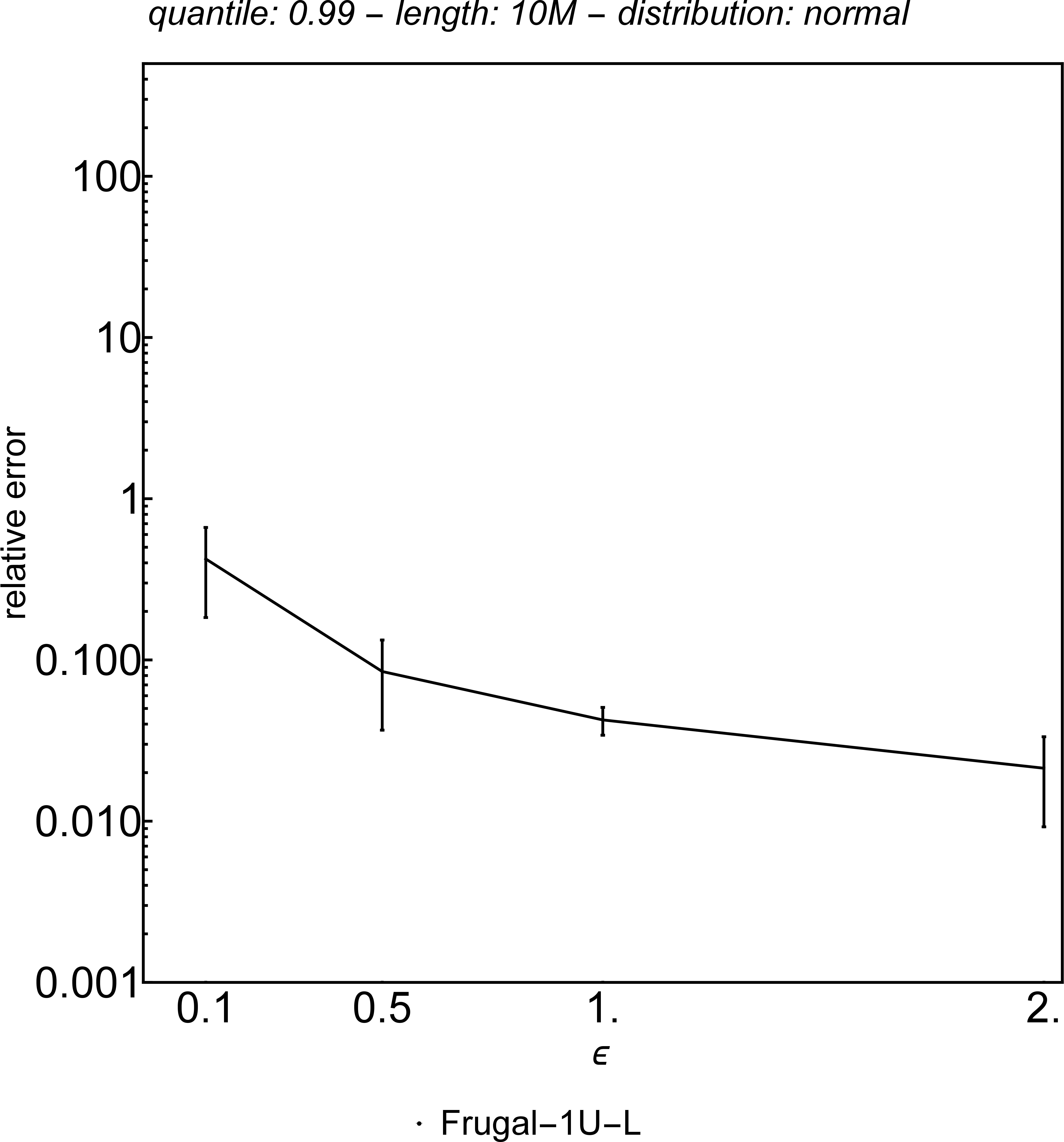}
        \caption{privacy budget}
        \label{fig:budget-all}
    \end{subfigure}
    \begin{subfigure}{0.33\textwidth}
        \includegraphics[width=0.99\linewidth]{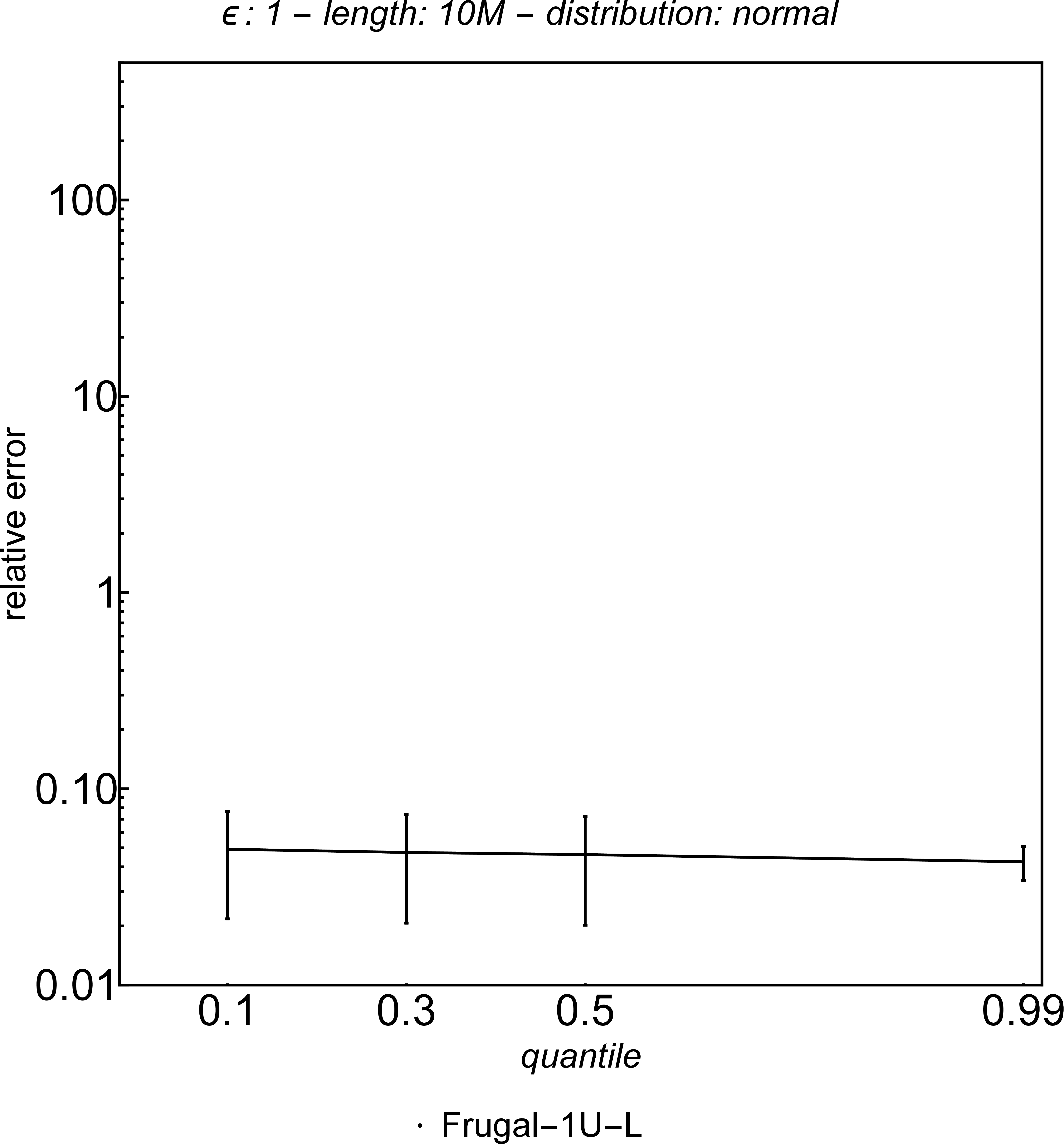}
        \caption{quantile}
        \label{fig:q-all}
    \end{subfigure}
    \begin{subfigure}{0.33\textwidth}
        \includegraphics[width=0.99\linewidth]{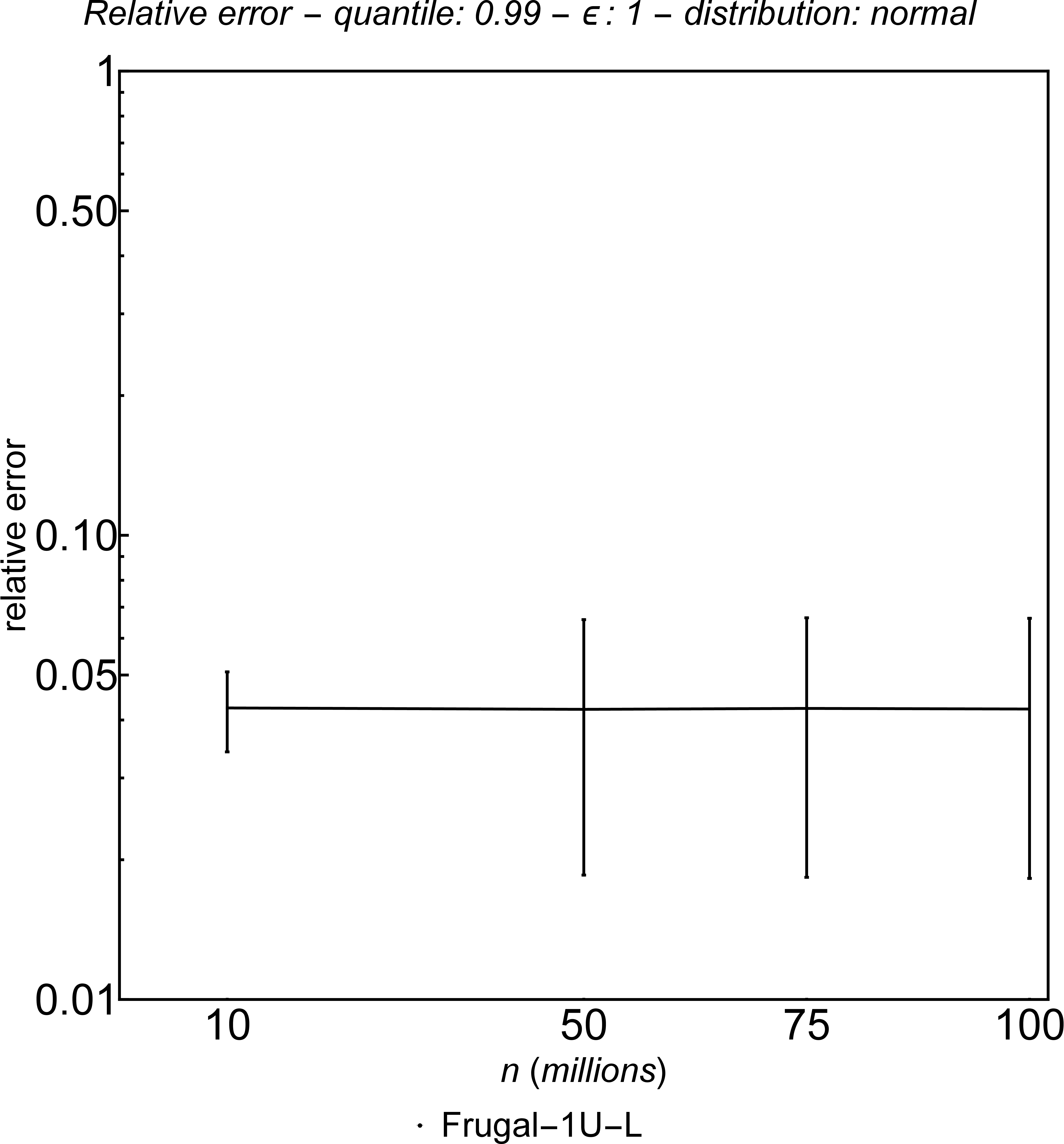}
        \caption{$n$}
        \label{fig:n-all}
    \end{subfigure}
    \begin{subfigure}{0.33\textwidth}
        \includegraphics[width=0.99\linewidth]{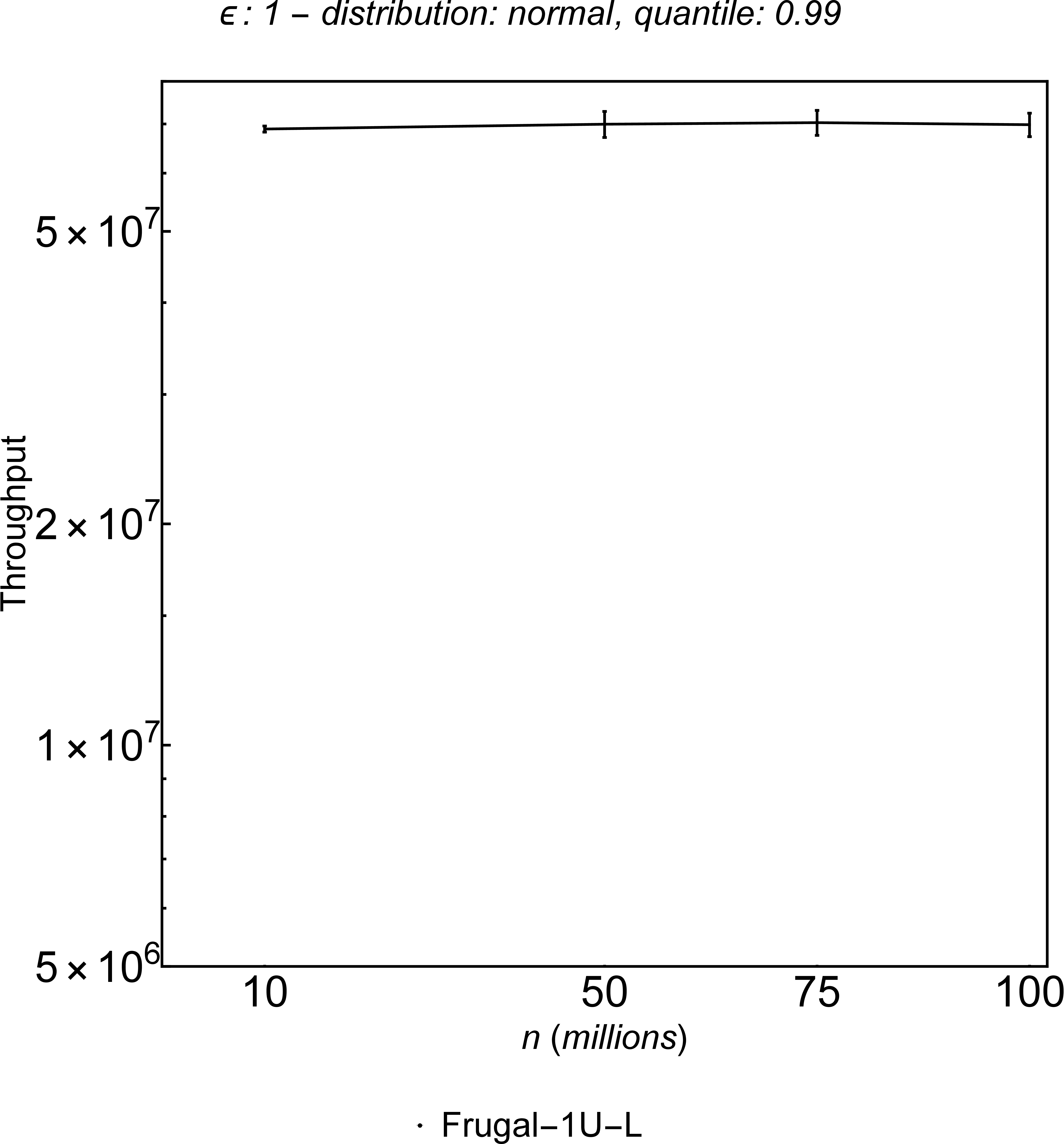}
        \caption{updates/s}
        \label{fig:throughput_all}
    \end{subfigure}
 
    \caption{\textsc{Frugal-1U-L}. Relative error varying the privacy budget $\epsilon$, the quantile $q$ and the stream size $n$. Throughput measured in updates/s.}
   \label{fig:1}
\end{figure*}

\begin{figure*}
    \centering
    \begin{subfigure}{0.33\textwidth}
        \includegraphics[width=0.99\linewidth]{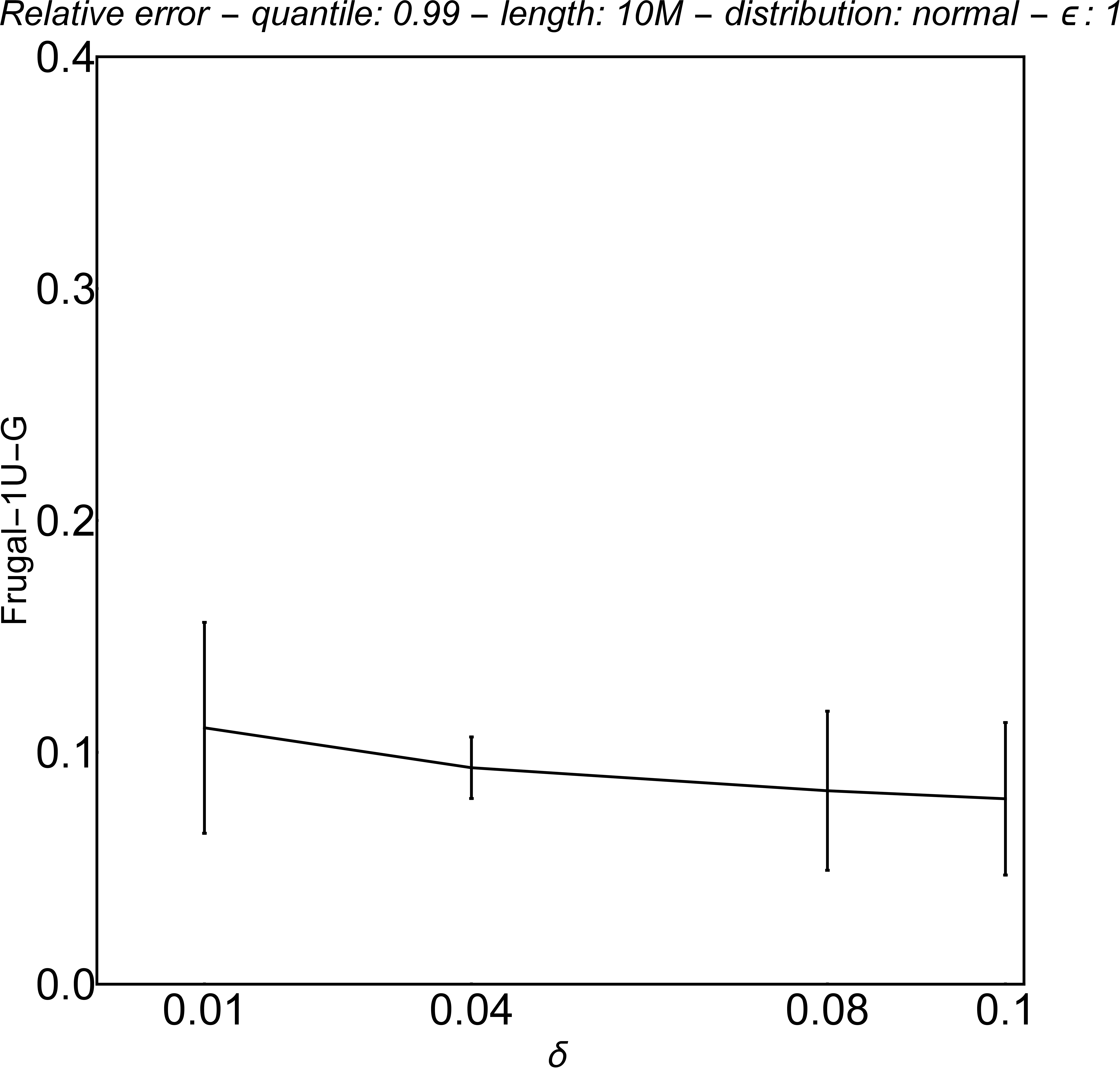}
        \caption{$\delta$}
        \label{fig:gauss-delta}
    \end{subfigure}\!\!
    \begin{subfigure}{0.33\textwidth}
        \includegraphics[width=0.99\linewidth]{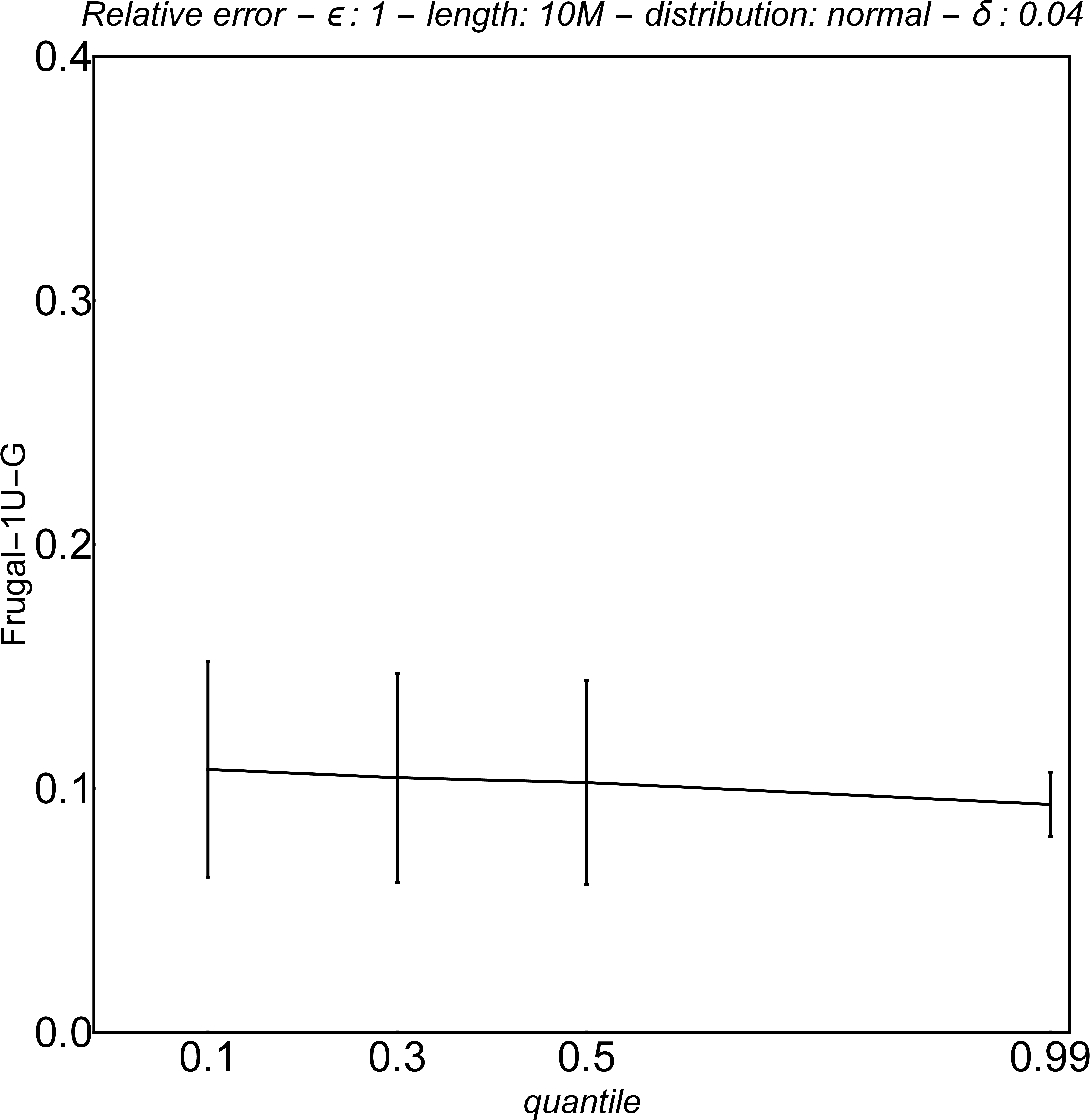}
        \caption{Quantile}
        \label{fig:gauss-q}
    \end{subfigure}\!\!
    \begin{subfigure}{0.33\textwidth}
        \includegraphics[width=0.99\linewidth]{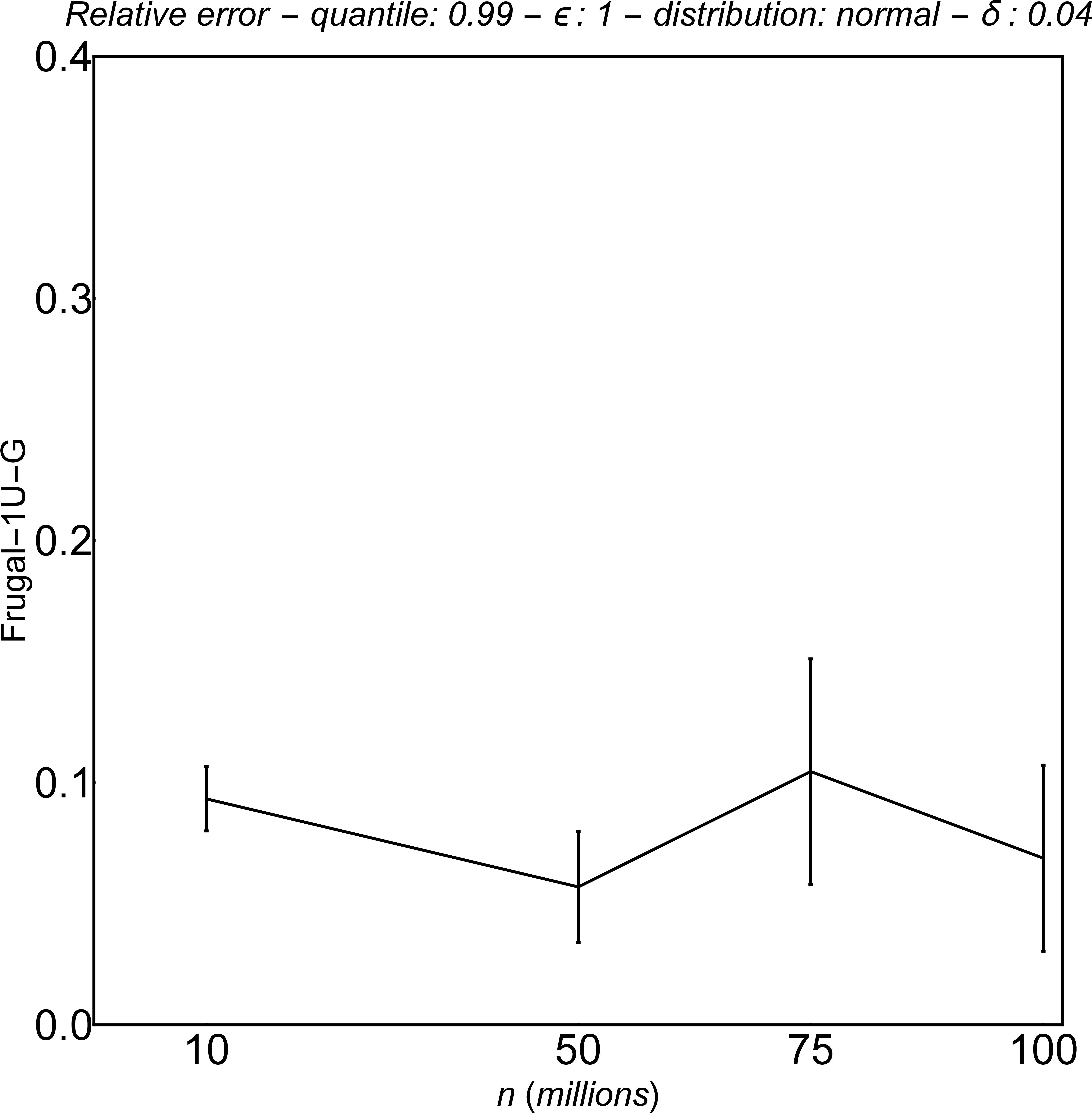}
        \caption{$n$}
        \label{fig:gauss-n}
    \end{subfigure}
 
   \caption{\textsc{Frugal-1U-G}. Relative error varying the probability $\delta$, the quantile $q$ and the stream size $n$.}
    \label{fig:gauss-1}
\end{figure*}

\begin{figure*}
    \centering
    \begin{subfigure}{0.33\textwidth}
        \includegraphics[width=0.99\linewidth]{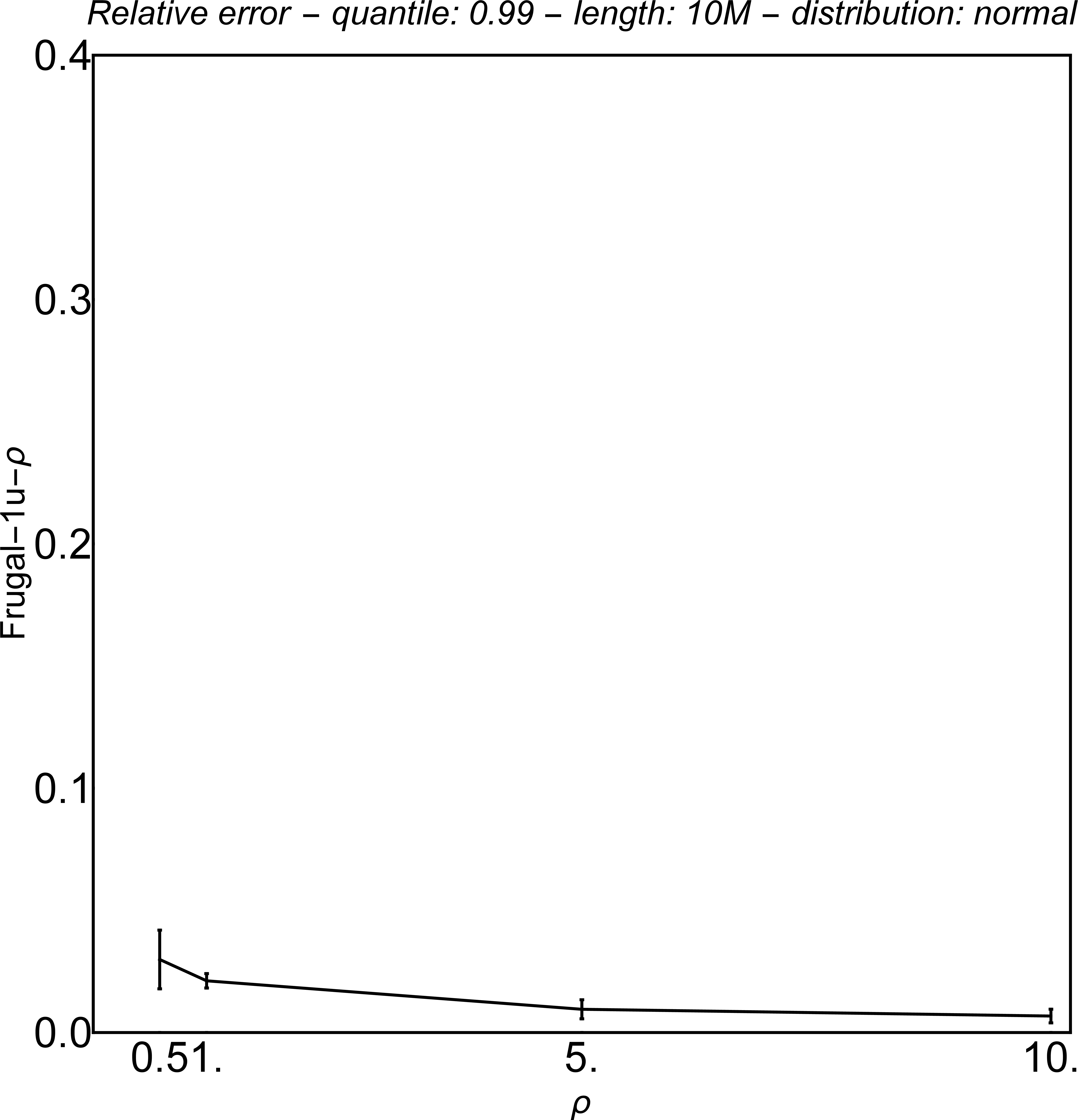}
        \caption{$\rho$}
        \label{fig:zcdp-rho}
    \end{subfigure}\!\!
    \begin{subfigure}{0.33\textwidth}
        \includegraphics[width=0.99\linewidth]{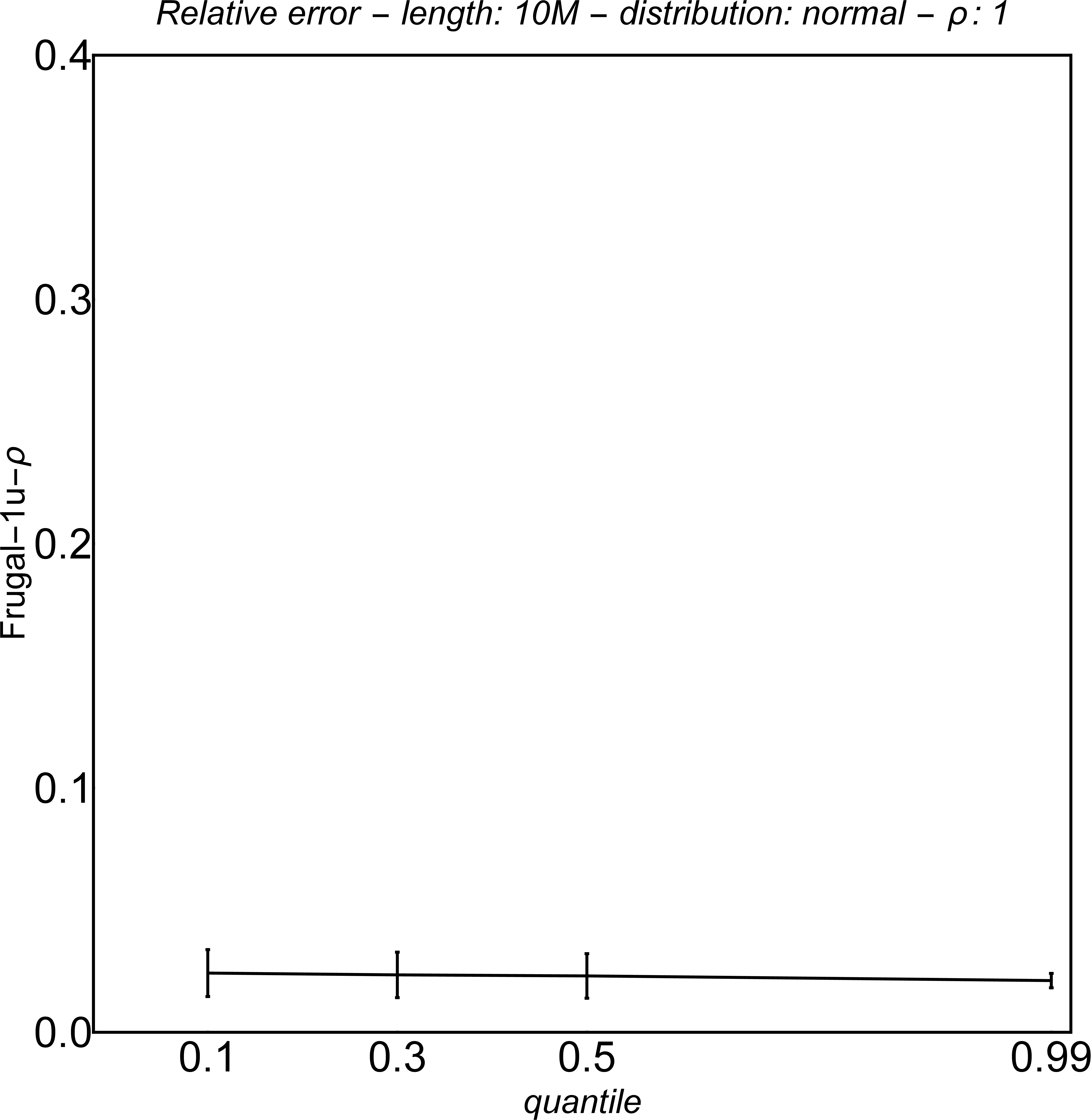}
        \caption{Quantile}
        \label{fig:zcdp-q}
    \end{subfigure}\!\!
    \begin{subfigure}{0.33\textwidth}
        \includegraphics[width=0.99\linewidth]{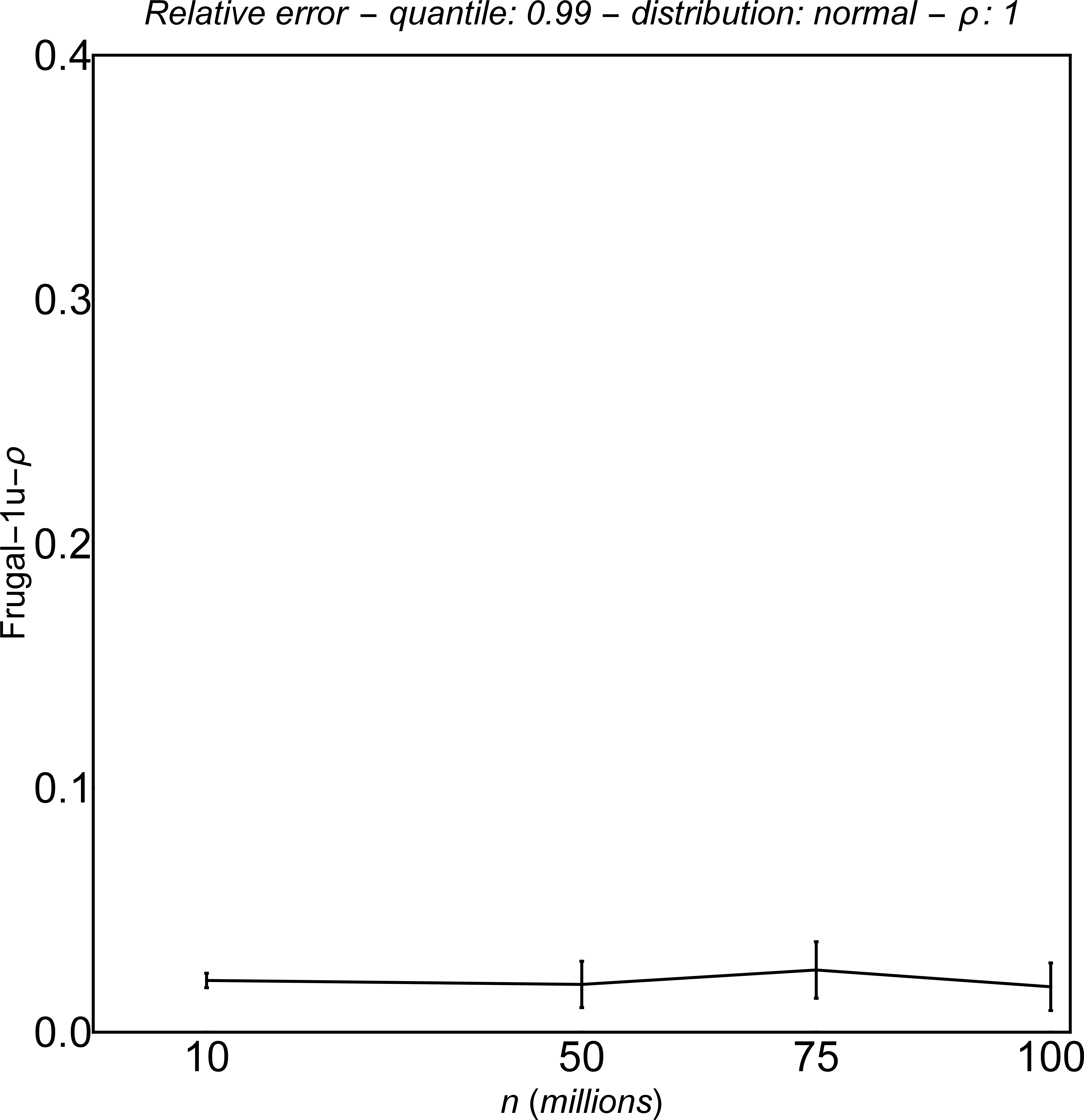}
        \caption{$n$}
        \label{fig:zcdp-n}
    \end{subfigure}
 
   \caption{\textsc{Frugal-1U-$\rho$}. Relative error varying the parameter $\rho$, the quantile $q$ and the stream size $n$.}
    \label{fig:zcdp-1}
\end{figure*}

\begin{figure*}
    \centering
        \includegraphics[width=0.85\linewidth]{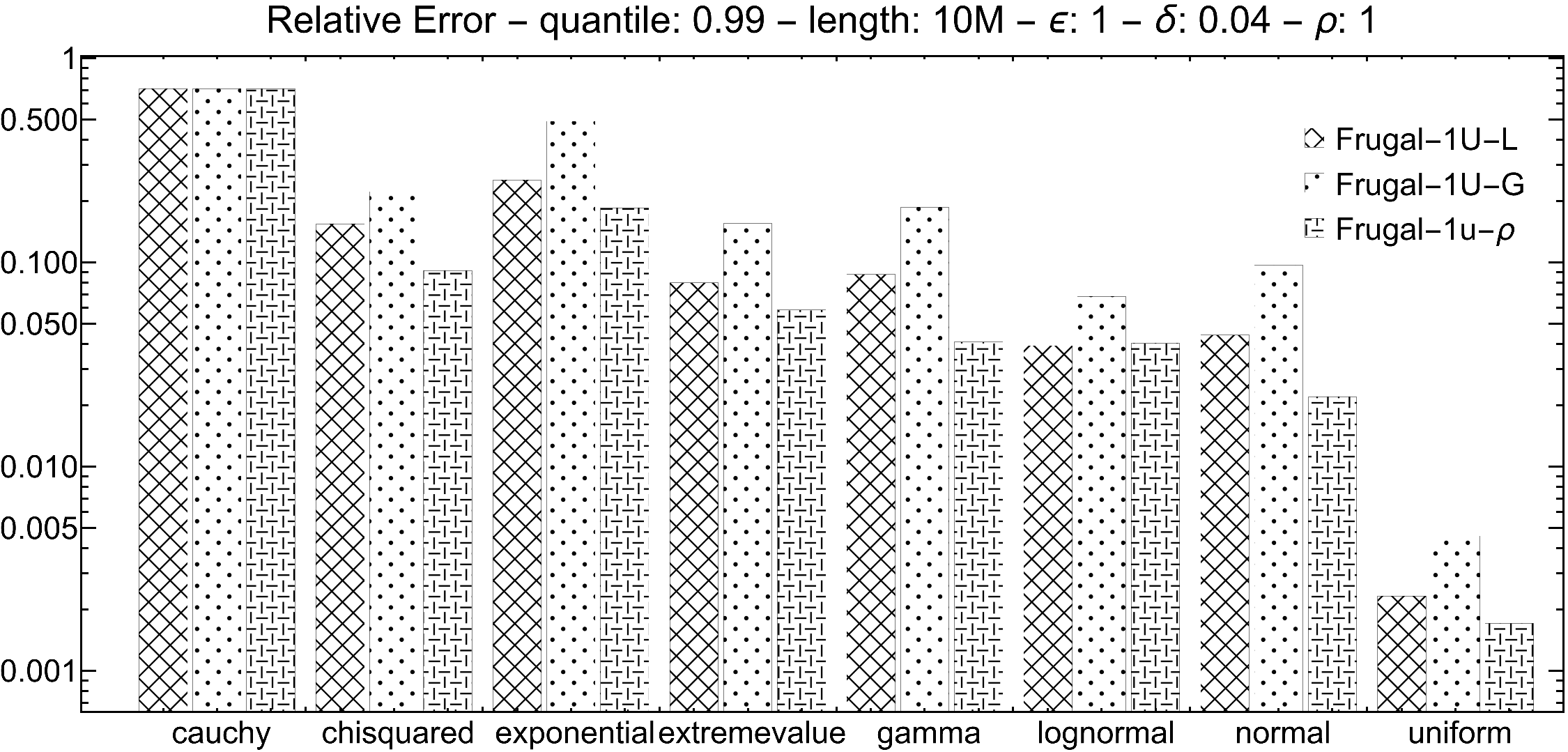}
        \caption{\textsc{Frugal-1U-L}: relative error varying the distributions.}
        \label{fig:distributions-all}
    
    \label{fig:zcdp-1}
\end{figure*}

We plot the relative error between the true quantile and the DP quantile estimate released by the algorithms under test, by allowing one parameter to vary whilst keeping the values of the others at their defaults. In all of the figures, the distribution used is the normal (later we compare the results obtained when varying the distribution as well). 

The experimental results for \textsc{Frugal-1U-L} (using the Laplace mechanism) are shown in Figure \ref{fig:1}. As depicted in Figure \ref{fig:budget-all}, the relative error decreases as expected when the privacy budget $\epsilon$ increases, meaning that the utility (see Section \ref{notation}) of the released value increases when $\epsilon$ increases. Therefore, a good tradeoff between privacy and utility is reached for $0.5 \leq \epsilon \leq 1$.

Figure \ref{fig:q-all} and Figure \ref{fig:n-all} depict the relative error varying respectively the computed quantile and the stream size, showing that the two quantities do not affect the security of the released quantile. Finally, Figure \ref{fig:throughput_all} depicts the throughput measured in updates/s.

Next, we analyze \textsc{Frugal-1U-G}. Increasing $\delta$, the probability of failure, provides as expected slightly less privacy, as shown in Figure \ref{fig:gauss-delta}. Varying the computed quantile exhibits a similar behaviour. In Figure \ref{fig:gauss-q}, slightly less privacy is associated to higher quantiles. Finally, the impact of the stream size is depicted in Figure \ref{fig:gauss-n}, in which a fluctuating behaviour can be observed, even though the interval of variation is tight.

Regarding \textsc{Frugal-1U-$\rho$}, Figure \ref{fig:zcdp-rho} shows that, as expected, increasing the privacy budget $\rho$ the relative error decreases and correspondingly the utility increases. A good tradeoff between privacy and utility is reached for $0.5 \leq \rho \leq 1$. Figure \ref{fig:zcdp-q} and \ref{fig:zcdp-n}, related respectively to the relative error varying the computed quantile and the stream size present the same behaviour illustrated for the Gaussian mechanism. This is not surprising, since this mechanism adds Gaussian noise (though the way noise is derived is of course different).

We now turn our attention to what happens when we vary the distribution. Figure \ref{fig:distributions-all} provides the results for \textsc{Frugal-1U-L}, \textsc{Frugal-1U-G} and \textsc{Frugal-1U-$\rho$}. As shown, the relative error between the true quantile and the DP quantile estimate released by one of the algorithms varies with the distribution. However, for our proposed algorithms, as expected (since the global sensitivity is just 2) the algorithms can be used independently of the actual distribution, with the notable exception related to the Cauchy distribution (which can be considered adversarial for our algorithms based on \textsc{Frugal-1U} as discussed in \cite{frugal}).

Our results show that, having fixed a distribution, the behaviour of our algorithms based on \textsc{Frugal-1U} does not depend on the seed used to initialize the pseudo-random number generator used to draw samples from the distribution. In this sense, our algorithms are robust. 

Finally, we analyze the $(\alpha, \beta)$-accuracy (Definition \ref{accuracy}) of \textsc{Frugal-1U-L}. Fixing $\beta = 0.04$, $\epsilon = 1$ and taking into account that the global sensitivity of \textsc{Frugal-1U} is $s = 2$, by using equation \eqref{alpha} we get $\alpha=\ln \left(\frac{1}{0.04}\right) \cdot 2 = 6.4$, so that \textsc{Frugal-1U-L} is (6.4, 0.04)-accurate.

For \textsc{Frugal-1U-G}, using Eq. \eqref{alpha-gauss} with $\delta = 0.04$, $\beta = 0.04$ and $\epsilon = 1$ we get $\alpha = 9.1$ so that \textsc{Frugal-1U-G} is (9.1, 0.04)-accurate. Finally, \textsc{Frugal-1U-$\rho$} accuracy is determined by using equation \eqref{alpha-rho} with $\rho = 1$ and $\beta = 0.04$, so that $\alpha = 2.4$ and  \textsc{Frugal-1U-$\rho$} is (2.4, 0.04)-accurate.

\section{Conclusions}
\label{conclusions}
In this paper, we proposed DP algorithms for tracking quantiles in a streaming setting. Our algorithms are DP variants of the well-known \textsc{Frugal-1U}algorithm, characterized by the property of requiring just a tiny amount of memory to process a stream while guaranteeing surprising accuracy for the estimates of a quantile. In particular, for \textsc{Frugal-1U} we gave corresponding $\epsilon$-DP, $(\epsilon,\delta)$-DP, and $\rho$-zCDF algorithms after proving that the global sensitivity of \textsc{Frugal-1U} is equal to 2. Finally, we also showed that the proposed algorithms achieve good accuracy in the experimental results.

\section*{Declaration on Generative AI}
  The author(s) have not employed any Generative AI tools.
  \newline

\bibliography{biblio}

\end{document}